\documentclass{ifacconf}

\usepackage{graphicx}
\usepackage{epstopdf}

\usepackage{float}
\usepackage{mathtools}
\usepackage{makecell}

\newenvironment{proof}[1][Proof]{\textbf{#1.} }{\ \rule{0.5em}{0.5em}}
\newtheorem{assumption}{Assumption}
\usepackage{amssymb}
\usepackage{graphicx}
\usepackage{natbib}        
\makeatletter
\DeclareOldFontCommand{\rm}{\normalfont\rmfamily}{\mathrm}
\DeclareOldFontCommand{\sf}{\normalfont\sffamily}{\mathsf}
\DeclareOldFontCommand{\tt}{\normalfont\ttfamily}{\mathtt}
\DeclareOldFontCommand{\bf}{\normalfont\bfseries}{\mathbf}
\DeclareOldFontCommand{\it}{\normalfont\itshape}{\mathit}
\DeclareOldFontCommand{\sl}{\normalfont\slshape}{\@nomath\sl}
\DeclareOldFontCommand{\sc}{\normalfont\scshape}{\@nomath\sc}
\makeatother

\usepackage{csquotes}  

\usepackage{physics}   

\DeclareMathOperator{\ess}{ess}

\newcommand \diag  {\operatorname{diag}}

\newcommand \R   {\mathbb{R}}

\newcommand \K   {\mathcal{K}}
\newcommand \Kinf{\mathcal{K_\infty}}

\newcommand \KL  {\mathcal{KL}}

\newcommand \LL  {\mathcal{L}}

\newcommand{\vertiii}[1]{{\left\vert\kern-0.25ex\left\vert\kern-0.25ex\left\vert #1 
    \right\vert\kern-0.25ex\right\vert\kern-0.25ex\right\vert}}

\newcommand{\normt}[1]{{\left\vert\kern-0.25ex\left\vert\kern-0.25ex\left\vert #1 
		\right\vert\kern-0.25ex\right\vert\kern-0.25ex\right\vert}}

\newif\ifMath					
\newif\ifEngi					

\newif\ifDFGtext					 

\newif\ifAndo              
													
\newif\ifExercises					
\newif\ifSolutions          
\newif\ifGerman							
\newif\ifEnglish						

\newif\ifnothabil						

\newif\ifFuture							

\newif\ifConf                    
\newif\ifJournal								 

\newif\ifNOTFORBOOK
\newif\ifFullVersion
\newif\ifExludedDueToSpaceReasons

\usepackage{xifthen}

\newcommand{\einsnorm}[2]{\ensuremath{
    \!\!\;\!\!\!\;
    \left\bracevert\!\!\!\!\!\left\bracevert
    \!
		\ifthenelse{\isempty{#2}}{#1}{#1(#2)}
    \!
      \right\bracevert\!\!\!\!\!\right\bracevert
    \!\!\;\!\!\!\;
  }}

\usepackage{xcolor}
\definecolor{blond}{rgb}{0.98, 0.94, 0.75}
	
\newlength\mytemplen
\newsavebox\mytempbox

\makeatletter
\newcommand\mybluebox{%
    \@ifnextchar[
       {\@mybluebox}%
       {\@mybluebox[0pt]}}

\def\@mybluebox[#1]{%
    \@ifnextchar[
       {\@@mybluebox[#1]}%
       {\@@mybluebox[#1][0pt]}}

\def\@@mybluebox[#1][#2]#3{
    \sbox\mytempbox{#3}%
    \mytemplen\ht\mytempbox
    \advance\mytemplen #1\relax
    \ht\mytempbox\mytemplen
    \mytemplen\dp\mytempbox
    \advance\mytemplen #2\relax
    \dp\mytempbox\mytemplen
    \colorbox{blond}{\hspace{1em}\usebox{\mytempbox}\hspace{1em}}}

\makeatother

\makeatletter
\let\origd=\d
\renewcommand*\d{
  \relax\ifmmode
    \mathrm{d}%
  \else
    \expandafter\origd
  \fi
}\makeatother

\usepackage{mathtools}

\makeatletter 
\newcommand{\pushright}[1]{\ifmeasuring@#1\else\omit\hfill$\displaystyle#1$\fi\ignorespaces}
\newcommand{\pushleft}[1]{\ifmeasuring@#1\else\omit$\displaystyle#1$\hfill\fi\ignorespaces}
\makeatother 

\newcounter{syscounter}

\newcounter{WPcounter}
\newcounter{PRcounter}
\newcommand{\remarkend}{\hfill $\blacktriangle$} 
\begin{document}

\begin{frontmatter}

\title{Virtual Resistance-Based Control \\ for Grid-Connected Inverters\\ using Persidskii Systems Approach}

\author[Aix]{Chakib Chatri} 
\author[Inria]{Ajul Dinesh}
\author[Aix]{Moussa Labbadi}

\address[Aix]{Aix-Marseille University, LIS UMR CNRS 7020, 13013 Marseille, France 
  (e-mail: \{chakib.chatri, moussa.labbadi\}@lis-lab.fr).}

\address[Inria]{Inria, Univ. Lille, CNRS, UMR 9189 - CRIStAL, 59000 Lille, France.
  (e-mail: ajul.dinesh@inria.fr).}

\begin{abstract}
This work addresses virtual resistance (VR)–based control for grid-connected inverters, which enhances transient damping, reduces steady-state errors, and improves robustness to grid disturbances without requiring additional voltage sensors. Classical passivity-based VR control is robust, but limited by restrictive sector bounds on nonlinearities. We extend these bounds and model the closed-loop system as a generalized Persidskii-type nonlinear system. Using this framework, we derive input-to-state stability (ISS) conditions that account for the extended nonlinearities and external disturbances, providing a systematic and less conservative approach to VR control design under practical operating conditions, which is validated through extensive simulations.
\end{abstract}


\end{frontmatter}

\section{Introduction}
Recently, virtual resistance (VR)–based control strategies have gained significant attention as a promising solution for power electronic control in modern electrical grids \citep{zhu2024active,chen2023modified}. A key advantage of VR-based methods is that they do not require additional voltage sensing hardware, making them attractive for scalable deployment in inverter-based systems \citep{Wang2014,Deng2021}. When applied to grid-connected converters, virtual resistance has been shown to improve stability by increasing damping, mitigating transient oscillations, reducing steady-state errors, and protecting power electronic devices from excessive load currents \citep{Me2021,Li2021}. These features make VR control a practical and intuitive alternative to synthetic inertia approaches for enhancing grid resilience, particularly in systems with a high penetration of power electronics.

Several recent studies have explored the potential of linear VR-based control laws for improving performance of grid-connected inverters. For instance, \cite{Me2021} employed a linear VR control strategy to enhance post-fault oscillation damping in grid-forming inverters. Similarly, \cite{Li2021} investigated linear VR-based control for improving the stability of grid-connected inverters in DC microgrids. Passivity-based linear VR controllers have also been used to analyze the stability of current-controlled grid-connected converters \citep{Xie2020, Valderrama2001}. Furthermore, linear VR control has been applied to DC solid-state transformers for suppressing circulating and harmonic currents \citep{Meng2021}. In addition, the influence of virtual resistance control on the transient stability of virtual synchronous generators has been investigated by \cite{Xiong2021}.

Despite these advances, the performance of linear VR controllers can significantly deteriorates under abrupt grid-voltage disturbances or gradual variations in grid parameters such as resistance and inductance \citep{anubi2022robust}. Such variations are increasingly common in modern distribution networks with high renewable penetration, nonlinear loads, and reconfigurable grid topologies.

To overcome these limitations, nonlinear VR control techniques based on passivity have been introduced for inverter-based systems \citep{gao2024passivity, anubi2022robust}. While these methods offer enhanced robustness and retain useful physics-based controller intuition, they typically impose restrictive sector-bound conditions on the admissible nonlinearities. These constraints limit the allowable virtual damping characteristics and often lead to conservative or technically complex stability conditions.

To remove these restrictions and obtain more tractable stability conditions, this work proposes a flexible VR-based control framework for grid-connected inverters. The proposed approach accommodates a broader class of current-dependent nonlinear VR controllers, while preserving essential passivity properties. We first establish input-to-state stability (ISS) of the closed-loop current-error dynamics of the grid-connected inverter using a Lyapunov-based analysis. Although these obtained conventional ISS conditions are stringent, we further reformulate the dynamics as a generalized Persidskii-type nonlinear system, enabling a relaxation of the obtained ISS conditions. By leveraging Persidskii-system analysis, we derive ISS conditions that explicitly account for the extended nonlinearities and external perturbations by employing a Lyapunov function that depends on the nonlinear structure of the system, thereby relaxing the sector bounds. This yields simplified and computationally tractable stability conditions suitable for inverter-based grid applications.

The main contributions of this work are as follows:
\begin{enumerate}
\item A flexible nonlinear virtual resistance–based control strategy is proposed, supporting a broader class of nonlinearities to enhance transient and steady-state performance in grid-connected inverters.
\item ISS of the closed-loop current dynamics is established using Lyapunov-based analysis, ensuring robustness to grid-voltage disturbances and grid-parameter uncertainties.
\item Through a Persidskii-type reformulation of the inverter current dynamics, simplified ISS conditions are derived in the form of linear matrix inequalities (LMIs), enabling tractable robustness assessment under external perturbations.
\end{enumerate}

The effectiveness of the proposed nonlinear VR-control strategy is demonstrated through numerical simulations involving significant grid-voltage variations and model-parameter perturbations.

The rest of the paper is organized as follows. Section~\ref{sec:preliminaries} reviews preliminaries on ISS and its Lyapunov characterizations. Section~\ref{sec:system_problem} describes the system dynamics and control objectives. Section~\ref{sec:main_results} presents the ISS analysis of the VR-controlled inverter using the Persidskii-system framework. Section~\ref{sec:simulation} provides numerical studies and comparisons. Finally, Section~\ref{sec:conclusion} concludes the paper.

\textit{Notations}: The symbols $\mathbb{R}$, $\mathbb{R}^n$, and $\mathbb{R}^{m \times n}$ denote the set of real numbers, the space of $n$-dimensional real vectors, and the space of $m \times n$ real matrices, respectively. Likewise, $\mathbb{R}_+$, $\mathbb{R}^n_+$, and $\mathbb{R}^{m \times n}_+$ represent the sets of non-negative real numbers, vectors, and matrices. For a vector $x \in \mathbb{R}^n$, $\|x\|$ denotes the Euclidean norm. For a Lebesgue-measurable function $u \in \mathbb{R}^m$, we define $\|u\|_{\infty}:= \ess \sup \{\|u(t)\|, ~ t \ge 0 \}$. The space of functions with $\|u\|_{\infty} < \infty$ is denoted by $\mathcal{L}_\infty^m$. The symbol $I$ denotes an identity matrix of appropriate dimensions. For a continuously differentiable function $V: \mathbb{R}^n \to \mathbb{R}$, the expression $\nabla V(x)^\top f(x) $ denotes the Lie derivative of $V$ along $f$, evaluated at $x \in \mathbb{R}^n$.

\section{Preliminaries}\label{sec:preliminaries}
In this section, we review several fundamental notions, including comparison function classes, ISS concepts, and Lyapunov-based characterizations of ISS.
\subsection{Comparison Functions}
To formalize various stability properties, we begin by introducing the standard classes of comparison functions:
\begin{equation*} \begin{array}{ll} {\K} &:= \left\{\gamma:\R_+\rightarrow\R_+\left|\ \gamma\mbox{ is continuous, strictly} \right. \right. \\ &\phantom{aaaaaaaaaaaaaaaaaaa}\left. \mbox{ increasing and } \gamma(0)=0 \right\}, \\ {\K_{\infty}}&:=\left\{\gamma\in\K\left|\ \gamma\mbox{ is unbounded}\right.\right\},\\ {\LL}&:=\left\{\gamma:\R_+\rightarrow\R_+\left|\ \gamma\mbox{ is continuous and strictly}\right.\right.\\ &\phantom{aaaaaaaaaaaaaaaa} \text{decreasing with } \lim\limits_{t\rightarrow\infty}\gamma(t)=0\},\\ {\KL} &:= \left\{\beta:\R_+\times\R_+\rightarrow\R_+\left|\ \beta \mbox{ is continuous,}\right.\right.\\ &\phantom{aaaaaa}\left.\beta(\cdot,t)\in{\K},\ \beta(r,\cdot)\in {\LL},\ \forall t\geq 0,\ \forall r >0\right\}. \\ \end{array} \end{equation*}
For a comprehensive presentation of these function classes, the reader is referred to \citep[Appendix A]{Mir23} and \citep{Kel14}.

\subsection{Stability Definitions}
Consider the nonlinear control system,
\begin{equation}
\label{xdot=f_xu}
\dot{x}(t) = f(x(t), d(t)), \quad x(0) = x_0 \in \R^n,
\end{equation}
where \(x(t) \in \R^n\) denotes the state and \(d(t) \in \mathcal{L}_\infty^m \subset \R^m\) denotes the control input/disturbance. The function \(f : \R^n \times \R^m \to \R^n\) is locally Lipschitz in \(x\), for every fixed \(d\).
We denote by \( x(t, x_0, d) \) the unique solution of \eqref{xdot=f_xu} at time \(t\), starting from the initial condition \(x_0\) under the input \(d\).
A key notion used throughout this paper is \emph{input-to-state stability (ISS)}, introduced by \cite{Son89}, which provides a formal framework for analyzing robust stability of nonlinear systems.

\begin{defn}
\label{Def:ISS}
System \eqref{xdot=f_xu} is called \emph{ISS}, if it is forward complete and there exist $\beta \in \KL$ and $\gamma \in \K$ such that for all $x_0 \in \R^n$, all $d \in \mathcal{L}_\infty^m$, and all $t\ge 0$, the following holds:
\begin{equation}
\label{iss_sum}
\|x(t,x_0,d)\| \leq \beta(\|x_0\|,t) + \gamma(\|d\|_{\infty}).
\end{equation}
\end{defn}

In \eqref{iss_sum}, the function $\beta$ describes the transient behavior of the ISS system, while $\gamma$ quantifies its asymptotic deviation from the origin. 

ISS systems exhibit several fundamental properties. For instance, by choosing \(d \equiv 0\) in \eqref{iss_sum}, 
the state trajectories are bounded by the \(\KL\)-function \(\beta\), which provides a uniform decay rate for all solutions starting within any bounded set.
Under our standing assumptions on \(f\), this property is equivalent to the \emph{global asymptotic stability} (i.e., local stability together with global attractivity) of the undisturbed system \citep{Mir23}.

On the other hand, taking the limit supremum of \eqref{iss_sum} as \( t \to \infty \), one obtains that if the system is ISS, then there exists a function \( \gamma \in \Kinf \) such that
\begin{equation}
\limsup_{t\to\infty} \|x(t, x_0, d)\| \le \gamma(\|d\|_{\infty}),
\quad x \in \R^n,\; d \in \mathcal{L}_\infty^m.
\label{eq:AG}
\end{equation}
This characteristic is known as the \emph{asymptotic gain property}, and the function \( \gamma \) is referred to as an \emph{asymptotic gain} of the system \eqref{iss_sum}.
In particular, every system with the asymptotic gain property (and thus every ISS system) has trajectories that remain bounded whenever the input magnitude is bounded.

\subsection{Lyapunov Characterizations}
This subsection presents the classical Lyapunov-based characterizations used to assess input-to-state stability. These results form the foundation for the analysis and design methodologies developed in the following sections.
\begin{defn}{\citep{Son89}}
\label{def:ISS_Lyap}
A smooth function \( V : \R^n \to \R_+ \) is called an \emph{ISS-Lyapunov function} for system \eqref{xdot=f_xu} if there exist comparison functions \( \gamma_1, \gamma_2, \gamma_3 \in \K_\infty \) and \( \alpha \in \K \) such that, for all \( x \in \R^n \) and \( u \in \R^m \),
\begin{align}
\gamma_1(\|x\|) &\le V(x) \le \gamma_2(\|x\|), \label{ISS_Lyap_bound} \\
\|x\| \ge \alpha(\|d\|_\infty) &\implies \nabla V(x) \cdot f(x,d) \le -\gamma_3(\|x\|). 
\end{align}
\end{defn}
\begin{thm} \citep{Son89} \label{thm:ISS_Lyap}
The system \eqref{xdot=f_xu} is ISS \textit{if and only if} it possesses an ISS-Lyapunov function.
\end{thm}
Therefore, establishing ISS reduces to constructing a function that satisfies the conditions stated in Definition~\ref{def:ISS_Lyap}.

\section{System Dynamics and Problem Formulation}\label{sec:system_problem}

\subsection{Grid-Connected Current Dynamics in the $dq$ Frame}
Consider a three-phase inverter connected to the grid through a resistor-inductor interface $r_g-l_g$, as shown in Fig.~\ref{fig:3ph_inverter},
where $r_g>0$ and $l_g>0$. 
\begin{figure} [!ht]
    \centering
     \includegraphics[width=0.9\linewidth]{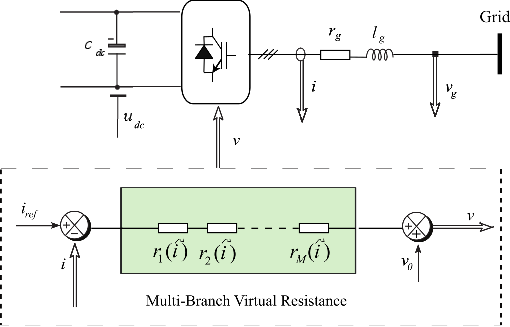}
    \caption{Multi-branch virtual resistance control scheme for the grid-connected inverter.}
    \label{fig:3ph_inverter}
\end{figure}

In the synchronous $dq$ frame, the current dynamics of the grid-connected inverter \citep{anubi2022robust} can be written as:
\begin{equation}
 \dot{i} = -\frac{1}{l_g}\big(r_g I - l_g W\big)\, i + \frac{1}{l_g} v - \frac{1}{l_g} v_g,
\label{eq:model}
\end{equation}
where $i=[i_d\; i_q]^\top\!\in\mathbb{R}^2$ is the line current, $v=[v_d\; v_q]^\top$ is the inverter voltage,
$v_g=[v_{gd}\; v_{gq}]^\top$ the grid voltage, and 
$W=\begin{bmatrix} 0 & \omega_g \\ -\omega_g & 0\end{bmatrix}$ is the skew-symmetric matrix with grid frequency $\omega_g$.

This model captures the coupling between the $d-q$ axes induced by grid frequency and provides a convenient form for analysis and control design to shape the grid power flow and current transients under grid disturbances.

\subsection{Control Objective}

The main objective of this work is to regulate the line current of a grid-connected inverter under abrupt grid-voltage variations and parameter variations in the system. To achieve this, we embed a virtual resistance control action into the inverter voltage $ v$, implemented using a multi-branch virtual resistance, as illustrated in Fig.~\ref{fig:3ph_inverter}.

Given a reference current $i_{ref}$, define the current error $\tilde i = i - i_{ref}$. Accordingly, for current regulation in the dynamics \eqref{eq:model}, we adopt a VR control law of the form 
\begin{equation}\label{eq:control_law}
v = v_0 - r(\tilde i)
\end{equation}
where the term $r(\tilde{i})$ provides a virtual voltage action with the bank of resistors. In \eqref{eq:control_law}, the nominal feedforward voltage $v_0$ is given by:
\begin{equation}
v_0 = (r_g I - l_g W)\, i_{ref} + v_{g_{ref}},
\label{eq:v0}
\end{equation}
where $v_{g_{ref}}$ denotes the nominal grid voltage. 

\begin{rem}
This VR-based control structure allows the inverter to dynamically regulate its damping in response to disturbances, improving transient performance in the grid without major architectural changes. Accordingly, it enhances robustness to voltage variations while preserving the simplicity of a standard current-controlled inverter. Moreover, using current-error feedback ensures consistent performance even under variations in $(r_g,l_g)$ arising from operating conditions or parameter uncertainty.
\remarkend
\end{rem}

\subsection{Multi-Branch Virtual Resistance Control}
While conventional VR control schemes rely on a single virtual branch, their ability to simultaneously shape small-signal damping and large-signal transient behavior is limited. To overcome this, in this paper, we introduce a multi-branch VR structure that enables finer control over nonlinear dissipation characteristics and improves resilience to large disturbances.

Accordingly, in the VR control law \eqref{eq:control_law} we define virtual voltage as, 
\begin{equation}
r(\tilde i) = \sum_{k=1}^{M} r_k(\tilde i)
\label{eq:vr_sum}
\end{equation}
where each branch mapping $r_k:\mathbb{R}^2\!\to\!\mathbb{R}^2$ is realized as a series interconnection of $m$ virtual resistance elements: 
\begin{equation}
r_k(\tilde i) = \sum_{l=1}^{m} r_k^l(\tilde i), 
\qquad k \in \{1, \dots, M\}.
\label{eq:branch_series}
\end{equation}

Substituting~\eqref{eq:control_law}-\eqref{eq:branch_series} into the grid-current model~\eqref{eq:model}
yields the closed-loop error dynamics:
\begin{equation}
\dot{\tilde{i}} 
= A\tilde{i}
- \frac{1}{l_g}\sum_{k=1}^{M} \sum_{l=1}^{m} r_k^l(\tilde{i})
- \frac{1}{l_g}\tilde{v}_g.
\label{eq:error_dyn_generalized}
\end{equation}
where $A=- \frac{1}{l_g}\big(r_g I - l_g W\big)$. Here, the input \(\tilde{v}_g\) acts as a perturbation to the current error dynamics. 

Thus, the control goal is to design the virtual voltages $r(\tilde{i})$ such that the current error satisfies $ \tilde i \to 0$, and the dynamics \eqref{eq:error_dyn_generalized} remain ISS with respect to both grid-voltage disturbances $\tilde v_g = v_g - v_{g_{ref}}$ and parameter variations in $(r_g,l_g)$. 

\begin{rem}
A single virtual resistance control scheme, as in \citep{anubi2022robust}, introduces only coarse damping and does not provide flexibility in adjusting both small-signal and large-signal behavior. In contrast, the proposed multi-branch virtual resistance configuration introduces additional degrees of freedom, allowing independent tuning of small-signal damping and large-signal transient response. 
\remarkend
\end{rem}

\section{Main Results}\label{sec:main_results}
In this section, we analyze the stability of the grid-connected inverter system with the proposed multi-branch VR control. Accordingly, we derive ISS conditions that ensure robustness of the current regulation in the presence of grid disturbances.

\subsection{ISS Analysis of Current Error Dynamics}
Initially, we present sufficient conditions under which the VR control law \eqref{eq:control_law} guarantees ISS of the closed-loop current error dynamics \eqref{eq:error_dyn_generalized} with respect to variations in the grid voltage.
\begin{thm}\label{thm:thm1}
Consider the closed-loop error dynamics in \eqref{eq:error_dyn_generalized}.
If there exists a continuously differentiable, positive-definite
function $V:\mathbb{R}^2 \to \mathbb{R}_+$, class $\mathcal{K}_\infty$ functions $\gamma_1$ and $\gamma_2$ and a scalar $\varepsilon>0$, such that, for all
$\tilde i \in \mathbb{R}^2$, the inequality:
\begin{gather}
\gamma_1(\|\tilde i \|) \le V(\tilde{i}) \le \gamma_2(\|\tilde i \|) \label{eq:condition1}\\
    \nabla  V(\tilde{i})^\top A\,\tilde i 
      - \frac{1}{l_g}\sum_{k=1}^{M} \sum_{l=1}^{m} \nabla V(\tilde{i})^\top r_k^l(\tilde{i})\nonumber \\
      \hspace{3cm}+\|\tilde i \|^2 +  \frac{\varepsilon}{2l_g}\|\nabla V(\tilde{i})\|^2\le 0 \label{eq:condition2}
\end{gather}
holds, then, the closed-loop error system \eqref{eq:error_dyn_generalized} with the VR control \eqref{eq:control_law} is
ISS with respect to grid voltage disturbances $\tilde{v}_g$.
\end{thm}

\begin{proof}
For the error dynamics \eqref{eq:error_dyn_generalized}, consider a continuously differentiable, radially unbounded, and positive definite Lyapunov function $V(\tilde{i})$, such that, 
\begin{equation*}
    \gamma_1(\|\tilde i \|) \le V(\tilde{i}) \le \gamma_2(\|\tilde i \|)
\end{equation*}
for $\gamma_1, \gamma_2 \in \mathcal{K}_\infty$. By differentiating $V(\tilde{i})$ along the trajectories of \eqref{eq:error_dyn_generalized}, we obtain,
\begin{equation}\label{eq:Lyapunov_e}
\begin{aligned}
      \dot V(\tilde{i})
    &= \nabla V(\tilde{i})^\top \dot{\tilde i}
    \nonumber\\[1mm]
    &= \nabla V(\tilde{i})^\top \Big(A\,\tilde i
      - \frac{1}{l_g}\sum_{k=1}^{M} \sum_{l=1}^{m} r_k^l(\tilde{i})
          - \frac{1}{l_g}\tilde v_g\Big)
       \nonumber\\[1mm]
    &= \nabla V(\tilde{i})^\top A\,\tilde i
      - \frac{1}{l_g}\sum_{k=1}^{M} \sum_{l=1}^{m} \nabla V(\tilde{i})^\top r_k^l(\tilde{i}) - \frac{1}{l_g}\nabla V(\tilde{i})^\top \tilde v_g  
\end{aligned}
\end{equation}
For an $l_g > 0$, since $- \frac{1}{l_g}\nabla V(\tilde{i})^\top \tilde v_g  \le \frac{1}{l_g}\big|\nabla V(\tilde{i})^\top \tilde v_g\big| $, we have, 
\begin{multline}\label{eq:dotV_bound}
    \dot V(\tilde{i})
    \le \nabla V(\tilde{i})^\top A\,\tilde i
      - \frac{1}{l_g}\sum_{k=1}^{M} \sum_{l=1}^{m} \nabla V(\tilde{i})^\top r_k^l(\tilde{i})\\
      + \frac{1}{l_g}\big|\nabla V(\tilde{i})^\top \tilde v_g\big|.
\end{multline}

Using Young's inequality, for any $\varepsilon>0$, we have the relation, 
\begin{equation*}
    \big|\nabla V(\tilde{i})^\top \tilde v_g\big| \le \|\nabla V(\tilde{i}) \|\|\tilde{v}_g \| \le \frac{\varepsilon}{2}\|\nabla V(\tilde{i}) \|^2 + \frac{1}{2 \varepsilon}\|\tilde{v}_g \|^2.
\end{equation*}
Therefore, \eqref{eq:dotV_bound} can be written as, 
\begin{align} \label{eq:dotV_pre}
    \dot V(\tilde{i})
    &\le  \nabla V(\tilde{i})^\top A\,\tilde i
      - \frac{1}{l_g}\sum_{k=1}^{M} \sum_{l=1}^{m} \nabla V(\tilde{i})^\top r_k^l(\tilde{i})  \nonumber\\
    &\quad
      +  \frac{\varepsilon}{2l_g}\|\nabla V(\tilde{i})\|^2
      +\frac{1}{2l_g\varepsilon}\|\tilde v_g\|^2
\end{align}

Introducing the term $\|\tilde i \|^2$, the above inequality can be equivalently written as,  
\begin{align}\label{eq:dotV_split}
    \dot V(\tilde{i})
    &\le \nabla V(\tilde{i})^\top A\,\tilde i - \|\tilde i \|^2
      - \frac{1}{l_g}\sum_{k=1}^{M} \sum_{l=1}^{m} \nabla V(\tilde{i})^\top r_k^l(\tilde{i}) \nonumber\\
    &\quad
      + \|\tilde i \|^2 +  \frac{\varepsilon}{2l_g}\|\nabla V(\tilde{i})\|^2
      +\frac{1}{2l_g\varepsilon}\|\tilde v_g\|^2.
\end{align}
For a $\varepsilon>0$, if the condition~\eqref{eq:condition2} 
is satisfied for all $\tilde{i}$, then the evolution of $V$ in \eqref{eq:dotV_split} can be expressed as:
\[
\dot V(\tilde{i}) \leq -\|\tilde i \|^2 +\frac{1}{2l_g\varepsilon}\|\tilde v_g\|^2.
\]
This implies ISS of the dynamics \eqref{eq:error_dyn_generalized} according to Definition \ref{def:ISS_Lyap} and Theorem \ref{thm:ISS_Lyap}, which completes the proof.
\end{proof}
\begin{rem}
The condition in \eqref{eq:condition2} provides a Lyapunov-based criterion to verify ISS of the error dynamics in \eqref{eq:error_dyn_generalized}. It reflects the combined effect of the linear part of the dynamics, characterized by the matrix \(A = - \frac{1}{l_g}(r_g I - l_g W)\), and the nonlinear contributions through the virtual voltage term \(\sum_{k=1}^{M} \sum_{l=1}^{m} r_k^l(\tilde{i})\). The term \(\frac{\varepsilon}{2l_g}\|\nabla V(\tilde{i})\|^2\) acts as a tunable damping component that compensates for potential destabilizing effects of the nonlinearities. Therefore, a properly chosen Lyapunov function \(V(\tilde{i})\)  ensures the stability of the generalized error dynamics despite the strongly nonlinear and multi-branch structure introduced by the virtual resistance network.
\remarkend
\end{rem}

Although Theorem~\ref{thm:thm1} provides a rigorous ISS condition, its direct application for analyzing the stability of the current error dynamics is challenging in practice. This is because the inequality \eqref{eq:condition2} involves quadratic (cross) terms in $\tilde{i}$ and $\nabla V(\tilde{i})$, both of which vary with time. Moreover, the multi-branch nonlinear resistor structure with virtual voltage term \(\sum_{k=1}^{M} \sum_{l=1}^{m} r_k^l(\tilde{i})\) 
renders classical quadratic Lyapunov functions insufficient for straightforward verification of \eqref{eq:condition2}.

Therefore, in the following, we present a simplified approach to establishing ISS for the dynamics \eqref{eq:error_dyn_generalized}. This is achieved by interpreting the system as a Persidskii-type nonlinear system \citep{persidskii1969concerning}, which enables the use of a composite Lyapunov function and leads to more tractable stability conditions in terms of linear matrix inequalities.

\subsection{Persidskii Representation and ISS Analysis}

A generalized state–space representation of a Persidskii system \citep{persidskii1969concerning} is given by,
\begin{equation}\label{eq:persidskii_system}
    \dot x(t)
    = A_0 x(t)
      + \sum_{k=1}^{M} A_k f_k\big(x(t)\big)
      + \varphi(t), \quad t \ge 0,
\end{equation}
where, $x(t)\in \mathbb{R}^n$ is the system state with $x(0)=x_0$. The matrices $A_s \in \mathbb{R}^{n \times n}$ for $s \in \{0, \dots, M\}$ are known constants, $f_k(x(t))=\begin{bmatrix}f_{k,1}(x_1(t)) &\dots &f_{k,n}(x_n(t))\end{bmatrix}^\top$ for $k \in \{1, \dots,M\}$ are nonlinear functions, and $\varphi(t) \in \mathcal{L}_\infty^n$ is a bounded disturbance input.

The current-error dynamics in \eqref{eq:error_dyn_generalized} can be written in Persidskii form \eqref{eq:persidskii_system} as, 
\begin{equation}\label{eq:current_persidskii}
    \dot{\tilde{i}} = A \tilde{i} + \sum_{k=1}^{M} - \frac{1}{l_g} r_k(\tilde{i})+-\frac{1}{l_g}\tilde{v}_g. 
\end{equation}
where, $ x(t) = \tilde{i}(t) \in \mathbb{R}^2$, $ A_0 = A \in \mathbb{R}^{2\times2}$, $A_k = -\tfrac{1}{l_g} I_2 \in \mathbb{R}^{2 \times 2}$, and $\varphi(t) = -\frac{1}{l_g}\tilde v_g(t) \in \mathbb{R}^2$. The nonlinear function $f_k(x(t))$ is, 
$$f_k(x) = r_k(\tilde{i}) =\sum_{l=1}^{m}r_k^l(\tilde{i})= \begin{bmatrix}r_{k,1}(\tilde{i}_d) \\ r_{k,2}(\tilde{i}_q)\end{bmatrix} \in \mathbb{R}^2.$$ 

To ensure the existence of solutions to \eqref{eq:current_persidskii}, we make the following assumption on the nonlinear functions $r_k(\tilde{i})$.
\begin{assumption}\label{assum:nonlinear_functions}
For any $k \in \{1, \dots,M\}$, $j \in \{1,2\}$ and $\tilde{i}_j \in \{\tilde{i}_d, \tilde{i}_q\}$ the nonlinearities 
$r_k^j(\tilde{i}_j)$ satisfy
\[
\tilde{i}_j\, r_k^j(\tilde{i}_j) > 0, \qquad \forall\, \tilde{i}_j \in \mathbb{R}\setminus\{0\}.
\]
\end{assumption}

Under Assumption~\ref{assum:nonlinear_functions}, there exists an index $p \in \{0,\dots,M\}$ such that for all $j \in \{1,2\}$ and $q \in \{1,\dots,p\}$,
\[
\lim_{\tilde{i}_j \to \pm\infty} r_q^j(\tilde{i}_j) = \pm\infty.
\]
Moreover, there exists $\mu \in \{p,\dots,M\}$ such that, for all 
$j \in \{1,2\}$ and $l \in \{1,\dots,\mu\}$,
\[
\lim_{\tilde{i}_j \to \pm\infty} 
\int_{0}^{\tilde{i}_j} r_l^j(s)\, ds = +\infty.
\]
In particular, the case $p = 0$  corresponds to the situation where all nonlinearities  are bounded.

With these preliminaries, we now provide simplified sufficient conditions, expressed as LMIs involving constant matrices, for ensuring ISS of the error dynamics \eqref{eq:current_persidskii} using the Persidskii–system framework.
\begin{thm}\label{thm:thm2}
    Consider the error dynamics in \eqref{eq:error_dyn_generalized} satisfying Assumption \ref{assum:nonlinear_functions}. If there exist symmetric matrices $P \succeq 0 \in \mathbb{R}^{2 \times 2}$, $\Phi \succeq 0 \in \mathbb{R}^{2 \times 2}$, and diagonal matrices $\Lambda_k = \diag(\Lambda_{k,1},\Lambda_{k,2}) \in \mathbb{R}^{2 \times 2}_+$ for $k\in \{1, \dots,M\}$, $\Omega_s \in \mathbb{R}^{2 \times 2}_+$ for $s \in \{0, \dots,M\}$, $\Upsilon_{s,l} \in \mathbb{R}^{2 \times 2}_+$ for $0 \le s < l \le M$, such that, 
    \begin{gather}
        P+\sum_{k=1}^{M}\Lambda_k \succ 0\label{LMI-P}\\
        \Omega_0 + \sum_{k=1}^{M}\Upsilon_{0,k}+\sum_{k=1}^{M}\Omega_k+ \sum_{s=1}^{M}\sum_{l=s+1}^{M}\Upsilon_{s,l} \succ 0\label{LMI-Xi}
    \end{gather}
   and \begin{equation}\label{LMI-Q}
       \Psi \preceq 0
   \end{equation} 
   with the block entries of $\Psi$ given by, 
   \begin{equation*}
       \begin{aligned}
          \Psi_{1,1}& = A^\top P+PA+\Omega_0, ~\Psi_{M+2,M+2} = \Phi,\\
          \Psi_{k+1,k+1} &= -2\frac{1}{l_g}\Lambda_k +\Omega_k, ~k \in \{1,\dots,M\},\\ 
          \Psi_{1,k+1}& = -\frac{1}{l_g}P + A^\top \Lambda_k+ \Upsilon_{0,k}, ~k \in \{1,\dots,M\},\\
            \Psi_{1,M+2}& = P,~  \Psi_{k+1,M+2} = \Lambda_k, ~ k \in \{1,\dots,M\},~ \text{and}\\
             \Psi_{s+1,l+1} &= -\frac{1}{l_g}\Lambda_l -\frac{1}{l_g}  \Lambda_l + \Upsilon_{s,l},~  ~s \in \{1,\dots,M-1\}, \\
             &\hspace{3cm}l \in \{s+1,\dots,M\}\\
       \end{aligned}
   \end{equation*}
   are satisfied, then the closed-loop error dynamics \eqref{eq:error_dyn_generalized} is ISS with respect to grid voltage perturbations $\tilde{v}_g$.
\end{thm}
\begin{proof}
For the dynamics \eqref{eq:error_dyn_generalized}, consider a Lyapunov candidate \citep{efimov2019robust, mei2022convergence} as, 
\begin{equation}\label{eq:lyap_fn_persidskii}
    V(\tilde{i}) = \tilde{i}^\top P \tilde{i} + 2 \sum_{k=1}^{M} \sum_{j=1}^2 \Lambda_{k,j}\int_{0}^{\tilde{i}_j} r_k^j(\tau) ~d\tau.
\end{equation}

Whenever \eqref{LMI-P} and Assumption~\ref{assum:nonlinear_functions} hold, the Lyapunov function \(V(\tilde{i})\) is positive definite and radially unbounded. Therefore, there exist class \(\mathcal{K}_\infty\) functions \(\gamma_1,\gamma_2\) such that,
\[
\gamma_1(\|\tilde i\|)\le V(\tilde i)\le \gamma_2(\|\tilde i\|),
\]
which verifies the first ISS Lyapunov condition in Definition \ref{def:ISS_Lyap}.

Differentiating $ V(\tilde{i})$ along the trajectories of \eqref{eq:error_dyn_generalized} gives, 
\begin{equation*}
\begin{aligned}
\dot{V}(\tilde{i}) &= \dot{\tilde{i}}^\top P \tilde{i} + \tilde{i}^\top P \dot{\tilde{i}} + 2 \sum_{k=1}^{M} \dot{\tilde{i}}^\top \Lambda_k r_k(\tilde{i})\\
&= \tilde{i}^\top\big(A^\top P + P A\big)\tilde{i}
- \tilde{i}^\top P\left(\frac{1}{l_g}\sum_{k=1}^{M} r_k(\tilde{i})\right)\\
&\quad
- \left(\frac{1}{l_g}\sum_{k=1}^{M} r_k(\tilde{i})\right)^\top P \tilde{i} - 2\,\tilde{i}^\top P\frac{1}{l_g}\tilde{v}_g \\
&\quad + 2\sum_{k=1}^M\Bigg(\tilde{i}^\top A^\top \Lambda_k r_k(\tilde{i})
- \frac{1}{l_g}\sum_{s=1}^{M} r_s(\tilde{i})^\top \Lambda_k r_k(\tilde{i})\\
&\quad
- \frac{1}{l_g}\tilde{v}_g^\top \Lambda_k r_k(\tilde{i})\Bigg).
\end{aligned}
\end{equation*}

Introducing non-negative matrices, $\Omega_0, \Omega_k, \Upsilon_{0,j}, \Upsilon_{s,r}, \Phi$, the expression can be equivalently written as, 
\begin{equation*}
    \begin{aligned}
    \dot{V}(\tilde{i}) & = z^\top \Psi z - \tilde{i}^\top \Omega_0 \tilde{i}  - \sum_{k=1}^{M}r_k(\tilde{i})^\top \Omega_k r_k(\tilde{i}) \\
    &\quad- 2 \sum_{k=1}^{M}  \tilde{i}^\top  \Upsilon_{0,k} r_k(\tilde{i}) -2 \sum_{s=1}^{M-1} \sum_{l=s+1}^{M} r_s(\tilde{i})^\top \Upsilon_{s,l} r_l(\tilde{i}) \\
    &\quad + \frac{1}{l_g^2}\tilde{v}_g^\top \Phi \tilde{v}_g
    \end{aligned}
\end{equation*}
where,  $z = \begin{bmatrix} \tilde{i} & r_1(\tilde{i}) & \dots &  r_M(\tilde{i})&\tilde{v}_g\end{bmatrix}^\top$.

Under \eqref{LMI-Q}, i.e., if $\Psi \preceq 0$, this yields the inequality, 
\begin{multline}\label{eq:v_bound_persidskii}
    \dot{V}(\tilde{i}) \le  - \tilde{i}^\top \Omega_0 \tilde{i}  - \sum_{k=1}^{M}r_k(\tilde{i})^\top \Omega_k r_k(\tilde{i})  - 2 \sum_{k=1}^{M}  \tilde{i}^\top  \Upsilon_{0,k} r_k(\tilde{i}) \\ -2 \sum_{s=1}^{M-1} \sum_{l=s+1}^{M} r_s(\tilde{i})^\top \Upsilon_{s,l} r_l(\tilde{i}) + \frac{1}{l_g^2}\tilde{v}_g^\top \Phi \tilde{v}_g.
\end{multline}

We further show that \(\dot V(\tilde{i})\) is upper bounded by a negative definite quadratic term with \(\tilde{i}\) and a term depending on the input \(\tilde v_g\), such that the above inequality satisfies the second condition for ISS in Definition \ref{def:ISS_Lyap}. 

To this end, consider inequality \eqref{LMI-Xi}, which guarantees that the diagonal matrices 
$\Omega_0 \succ 0$ and $\Omega_k \succ 0$ for $k=1,\dots,M$, and that the cross blocks 
$\Upsilon_{0,k}$ and $\Upsilon_{s,l}$ are finite. Next, each bilinear (cross) term can be 
bounded using Young's inequality. In particular, for any scalar $\varepsilon_k > 0$, $k \in 
\{1,\dots,M\}$, the term $2 \tilde{i}^\top \Upsilon_{0,k} r_k(\tilde{i})$ can be bounded as
\[
2 \tilde{i}^\top \Upsilon_{0,k} r_k(\tilde{i})
\le \varepsilon_k \tilde{i}^\top \tilde{i} + \frac{1}{\varepsilon_k}\, r_k(\tilde{i})^\top \Upsilon_{0,k}^\top \Upsilon_{0,k} \, r_k(\tilde{i}).
\]
Similarly, for \(s<l\) and any \(\delta_{s,l}>0\), we have,
\begin{equation*}
2 r_s(\tilde{i})^\top \Upsilon_{s,l} r_l(\tilde{i})
\le \delta_{s,l}\, r_s(\tilde{i})^\top r_s(\tilde{i}) + \frac{1}{\delta_{s,l}}\, r_l(\tilde{i})^\top \Upsilon_{s,l}^\top \Upsilon_{s,l}\, r_l(\tilde{i}).    
\end{equation*}

Now, upon choosing the scalar parameters \(\varepsilon_k\), $k\in \{1,\dots,M\}$ and \(\delta_{s,l}\), $s\in \{1,\dots,M-1\}$, $l\in \{s+1,\dots,M\}$ small enough,  the diagonal negative-definite terms \(-\tilde{i}^\top\Omega_0\tilde{i}\) and \(-r_k(\tilde{i})^\top\Omega_k r_k(\tilde{i})\) dominate the corresponding positive contributions from the right-hand side of above inequality bounds. Further, denote:
\[
\sigma_0 := \lambda_{\min}(\Omega_0)>0,\quad \sigma_k := \lambda_{\min}(\Omega_k)>0,\; k=1,\dots,M.
\]
Because \(\Upsilon_{0,k}^\top\Upsilon_{0,k}\) and \(\Upsilon_{s,l}^\top\Upsilon_{s,l}\) are finite, there exist choices of \(\varepsilon_k,\delta_{s,l}>0\) such that for each \(k\), we have,
\begin{gather*}
\sigma_0 - \varepsilon_k > 0,\qquad \sigma_k - \frac{1}{\varepsilon_k}\lambda_{\max}(\Upsilon_{0,k}^\top\Upsilon_{0,k}) \\ - \sum_{s\ne l}\frac{1}{\delta_{\min(s,l)}}\lambda_{\max}(\Upsilon_{\min(s,l),\max(s,l)}^\top\Upsilon_{\min(s,l),\max(s,l)})>0,
\end{gather*}
where the latter removes the contributions of cross-terms involving \(r_k(\tilde{i})\). (Such a choice is possible because each \(\sigma_k>0\).) With these choices, we obtain a constant \(\varsigma_1>0\) such that the quadratic terms in \eqref{eq:v_bound_persidskii} satisfy the lower bound
\begin{multline}\label{eq:bound_quadratic}
- \tilde{i}^\top \Omega_0 \tilde{i}  - \sum_{k=1}^{M} r_k(\tilde{i})^\top \Omega_k r_k(\tilde{i})  - 2 \sum_{k=1}^{M}  \tilde{i}^\top  \Upsilon_{0,k} r_k(\tilde{i}) 
\\ -2 \sum_{s=1}^{M-1} \sum_{l=s+1}^{M} r_s(\tilde{i})^\top \Upsilon_{s,l} r_l(\tilde{i}) 
\le - \varsigma_1 \|\zeta(\tilde{i})\|^2,
\end{multline}
where \(\zeta(\tilde{i})\) is a stacked vector collecting \(\tilde{i}\) and the \(r_k(\tilde{i})\)'s 
and \(\varsigma_1>0\) depends on the chosen \(\varepsilon_k,\delta_{s,l}\) and the minimal eigenvalues \(\sigma_0,\sigma_k\). Further, because \(r_k(\cdot)\) are functions of \(\tilde i\), there exists \(\varsigma_2>0\) (using norm equivalence in finite dimensions and the structure of the \(r_k\)) such that,
\[
\|\zeta(\tilde{i})\|^2 \ge \varsigma_2 \|\tilde{i}\|^2.
\]
Hence, the quadratic term \eqref{eq:bound_quadratic} in \eqref{eq:v_bound_persidskii} can be further bounded as,
\begin{multline*}
- \tilde{i}^\top \Omega_0 \tilde{i}  - \sum_{k=1}^{M} r_k(\tilde{i})^\top \Omega_k r_k(\tilde{i})  - 2 \sum_{k=1}^{M}  \tilde{i}^\top  \Upsilon_{0,k} r_k(\tilde{i}) 
\\ -2 \sum_{s=1}^{M-1} \sum_{l=s+1}^{M} r_s(\tilde{i})^\top \Upsilon_{s,l} r_l(\tilde{i})
\le - \varsigma\, \|\tilde{i}\|^2
\end{multline*}
for some constant \(\varsigma>0\).

The remaining term in \eqref{eq:v_bound_persidskii} depends only on the input, which can be bounded as, 
\begin{equation*}
   \begin{aligned}
   \frac{1}{l_g^2}\tilde{v}_g^\top \Phi \tilde{v}_g
&\le \frac{1}{l_g^2}\lambda_{\max}(\Phi)\,\|\tilde v_g\|^2 \\
&= \alpha\,\|\tilde v_g\|^2
   \end{aligned} 
\end{equation*}
where \(\alpha:=\tfrac{1}{l_g^2}\lambda_{\max}(\Phi)\ge 0\). Combining the two bounds yields
\[
\dot V(\tilde{i}) \le - \varsigma\,\|\tilde i\|^2 + \alpha\,\|\tilde v_g\|^2
\]
which corresponds to the second condition in Definition \ref{def:ISS_Lyap}.

Thus, the Lyapunov function in \eqref{eq:lyap_fn_persidskii} satisfies both the conditions in Definition \ref{def:ISS_Lyap}, whenever LMIs \eqref{LMI-P}-\eqref{LMI-Q} are satisfied, which implies the ISS of the closed loop error dynamics \eqref{eq:error_dyn_generalized} according to Theorem \ref{thm:ISS_Lyap}. 

Furthermore, using the above differential inequality and standard ISS Lyapunov arguments \citep{Son89}, we obtain the ISS estimate of the current error as, 
\[
\|\tilde i(t)\| \le \beta(\|\tilde i_0\|,t) + \gamma\big(\|\tilde v_g\|_{\infty}\big),
\]
for some \(\beta\in\mathcal{KL}\) and \(\gamma\in\mathcal{K}\), with \(\gamma(s)=\sqrt{\frac{\alpha}{\varsigma}}\;s\) (or any admissible \(\mathcal K\)-function upper bounding the steady-state gain). 
\end{proof}

\begin{rem}
The conditions in Theorem~\ref{thm:thm2} are obtained as LMIs involving only the system matrix $A$, which depends on the grid parameters $(r_g,l_g)$. Therefore, these conditions can be easily verified to analyze the ISS of the error dynamics \eqref{eq:error_dyn_generalized} under the proposed multi-branch VR control structure.  
\remarkend
\end{rem}
\begin{rem}
The obtained ISS conditions in Theorem~\ref{thm:thm2} may be further relaxed upon establishing the global stability of error system~(\ref{eq:error_dyn_generalized}). In particular, once global stability is ensured, all bounded nonlinearities $r_{k}(\tilde{i})$ can be incorporated into the input term $\tilde v_g$. Moreover, certain cross-terms involving bounded nonlinearities, such as
$\tilde v_g^\top \Lambda_j r_k(\tilde{i})$, can also be absorbed into $\tilde v_g$. This
reformulation can lead to a sparser structure for the matrix $\Psi$.
\remarkend
\end{rem}

\section{Simulation Results}\label{sec:simulation}
 In this section, we present simulation results for the proposed multi-branch virtual resistance–based control of the grid-connected inverter. The system parameters used in the simulations are summarized in
Table~\ref{tab:inverter-connected}.
\begin{table}[!h]
    \centering
    \caption{Grid and inverter parameters}
    \label{tab:inverter-connected}
    \begin{tabular}{lcc}
        \hline
        \textbf{Parameter} & \textbf{Symbol} & \textbf{Value} \\
        \hline
        Grid inductance & $l_g$ & $ 0.367\,\text{mH}$ \\
        Nominal grid resistance & $r_{gn}$ & $27.6\,\text{m}\Omega$ \\
        Nominal grid voltage (L--N, rms) & $ v_{g_{ref}}$ & $392\,\text{V}$ \\
        Grid frequency & $\omega_g$ & $60\,\text{Hz}$ \\
        \hline
    \end{tabular}
\end{table}

To evaluate the performance of the proposed controller, several test cases are considered, including grid-voltage variations and grid-parameter uncertainties. 
\subsection{Scenario 1: Performance under grid-voltage variations}
The first scenario is used to analyze the dynamic response and transient behavior of the system under grid-voltage fluctuations, by comparing the proposed method with the results presented in \cite{anubi2022robust}. 
A grid-voltage pulse of amplitude $+40\%\,v_{g_{ref}}$ is applied on the $d$-axis between $t=0.1~\text{s}$ and $t=0.101~\text{s}$, while the $q$-axis voltage remains equal to zero. The time responses of the currents $i_d$ and $i_q$ and their tracking errors $\tilde{i}_d$ and $\tilde{i}_q$ are shown in Fig.~\ref{fig:1} and Fig.~\ref{fig:2}, respectively, with the different VR laws (linear, cubic, hybrid, and $\sinh$) proposed in \cite{anubi2022robust} and with the proposed virtual resistance structure. 

\begin{figure}[!ht]
  \centering
\includegraphics[width=0.48\textwidth]{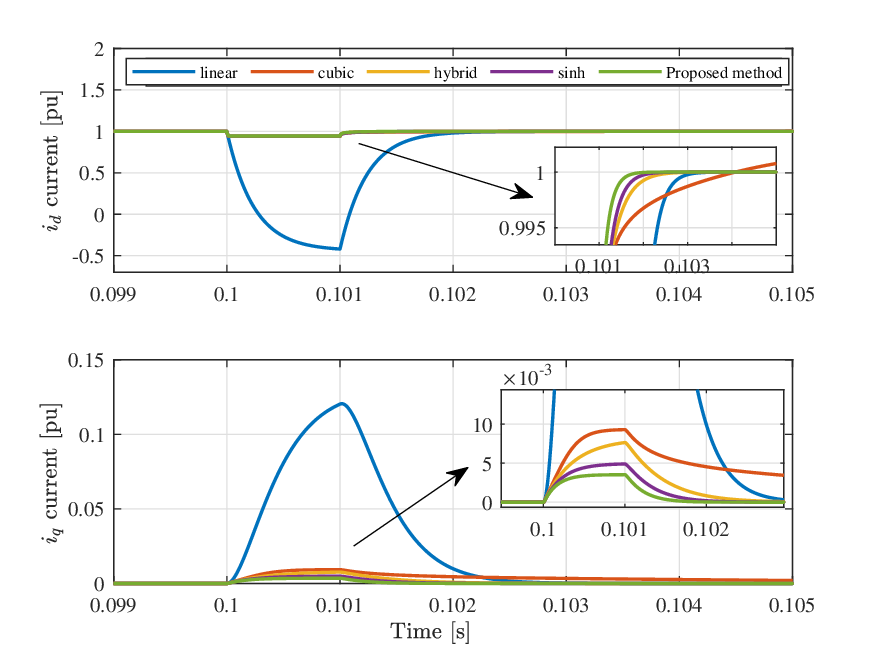}
  \caption{Response of current $i$ under grid-voltage variations}
  \label{fig:1}
\end{figure}
\begin{figure}[!ht]
  \centering
\includegraphics[width=0.48\textwidth]{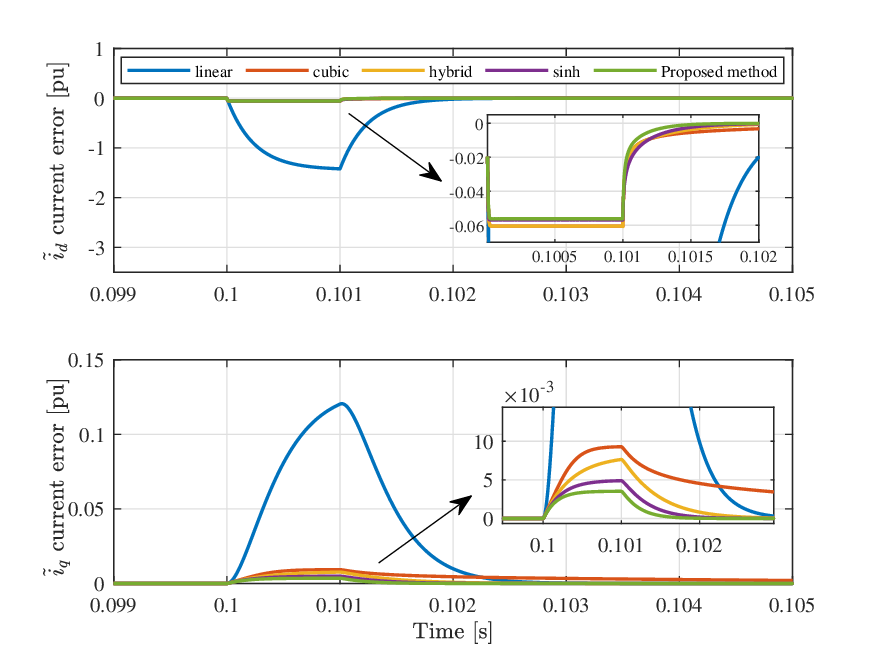}
  \caption{Response of current error $\tilde{i}$ under grid-voltage variations}
  \label{fig:2}
\end{figure}
It can be observed that the proposed method achieves the best current-tracking performance in both $i_d$ and $i_q$ among all single-branch VR control laws. In particular, the proposed multi-branch method significantly reduces the overshoot in $i_d$ after the voltage pulse and shortens the settling time compared with the single-branch resistances, while preserving a negligible steady-state error. Moreover, the proposed multi-branch structure reduces the overshoot in $i_q$ and provides a faster recovery after the voltage pulse. Because of the coupling with the $d$-axis voltage, $i_q$ shows a transient deviation, which is smaller and disappears faster with the proposed controller. Furthermore, the steady-state errors in both $i_d$ and $i_q$ remain very close to zero after the disturbance.

These observations are confirmed by the performance indices summarized in Table~\ref{tab:Performance1}, where the proposed method achieves the smallest settling times at $2\%$ and the lowest RMS values of the $d$-axis tracking error and the $q$-axis current error, when compared to the linear, cubic, hybrid, and $\sinh$ VR laws.
\begin{table}[!ht]
\centering
\caption{Performance indices for the different VR-control laws in Scenario 1.}
\label{tab:Performance1}
\footnotesize
\setlength{\tabcolsep}{5pt}
\begin{tabular}{lcccc}
\hline
VR control law &
\makecell[c]{Settling time\\$T_{s,d}^{2\%}$ [ms]} &
\makecell[c]{RMS ($\tilde{i}_d$)\\{}[A]} &
\makecell[c]{RMS ($\tilde{i}_q$)\\{}[A]} \\
\hline
Linear          & 2.000   &   12.6141 	& 1.0561\\
Cubic           & 1.040  &  0.6283	 & 0.1170\\
Hybrid          & 1.041  &  0.6243	 & 0.0706\\
$\sinh$         & 1.066   &   0.7808 & 0.0592\\
Proposed method & \textbf{1.03} & \textbf{0.5834} &  \textbf{0.0339}  \\
\hline
\end{tabular}
\end{table}

\subsection{Scenario 2: Performance under random grid-resistance variations}

In this scenario, we analyze the robustness of the proposed controller with respect to grid-parameter uncertainties. 
The grid resistance is modeled as a time-varying parameter $r_g(t)$ that randomly changes around its nominal value $r_{gn}$ as
$r_g(t) \in [0.1\,r_{gn},\, 1.9\,r_{gn}]$ for $t \in [0.2,\,0.8]~\text{s}.$
This case is used to compare the proposed multi-branch VR control with the single-branch linear, cubic, hybrid, and $\sinh$ virtual resistance laws under severe random variations of the grid resistance. 

The current responses $i_d$ and $i_q$ in Fig. \ref{fig:s2_1} and the current tracking errors $\tilde{i}_d$ and $\tilde{i}_q$ in Fig. \ref{fig:s2_2} show the effectiveness of the proposed method against grid-resistance variations, compared with the other single-branch virtual resistance laws. 
\begin{figure}[!ht]
  \centering
\includegraphics[width=0.48\textwidth]{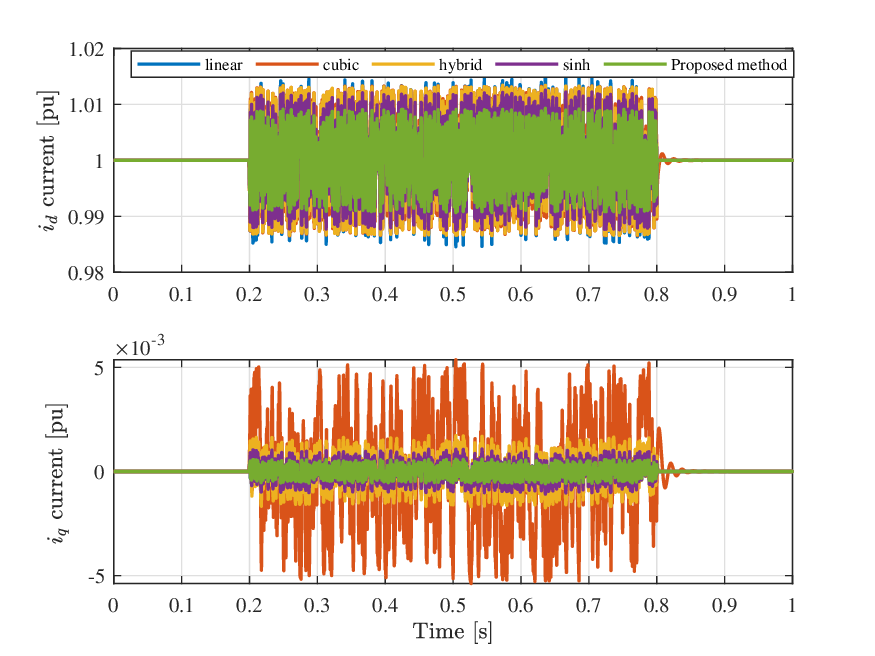}
  \caption{Response of currents under random grid resistance variations.}
  \label{fig:s2_1}
\end{figure}

\begin{figure}[!ht]
  \centering
\includegraphics[width=0.48\textwidth]{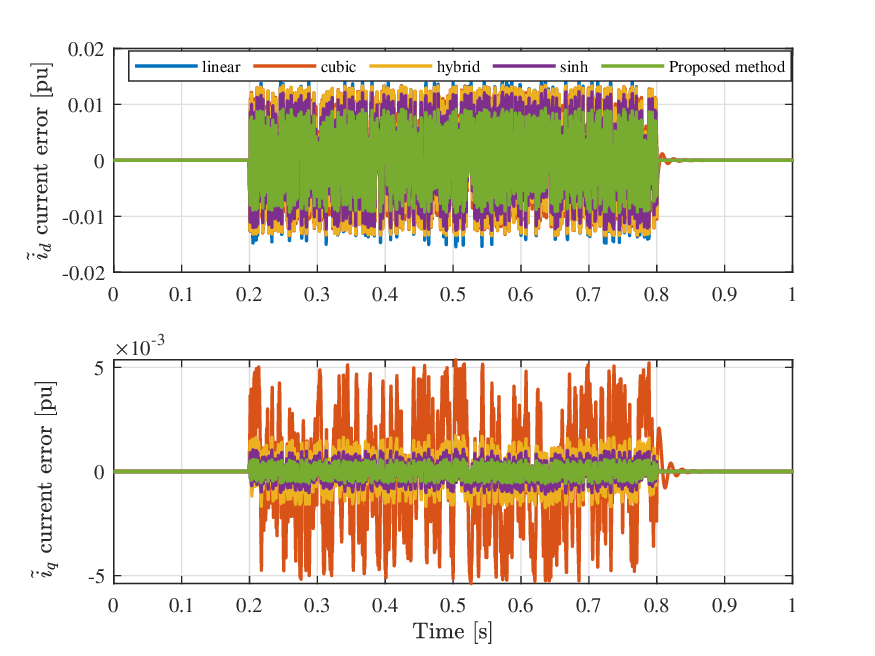}
  \caption{Response of currents errors under random grid resistance variations.}
  \label{fig:s2_2}
\end{figure}

Moreover, the performance indices in Table~\ref{tab:Performance2} show that the proposed controller offers the best compromise between tracking accuracy and robustness, as it consistently attains the smallest RMS values of the tracking errors.

\begin{table}[!ht]
\centering
\caption{Performance indices for the different VR-control laws in Scenario 2.}
\label{tab:Performance2}
\footnotesize
\setlength{\tabcolsep}{5pt}
\begin{tabular}{lcc}
\hline
VR control law &
\makecell[c]{RMS ($\tilde{i}_d$)\\ {}[A]} &
\makecell[c]{RMS ($\tilde{i}_q$)\\{}[A]} \\
\hline
Linear          & 0.6092  & 0.0477\\
Cubic           &0.7520	  & 0.2499 \\
Hybrid          & 0.7084  & 0.0762 \\
$\sinh$         & 0.5503 & 0.0425  \\
Proposed method & \textbf{0.4326} & \textbf{0.0252} \\
\hline
\end{tabular}
\end{table}
\section{Conclusion}\label{sec:conclusion}
In this paper, we propose a robust and flexible virtual resistance–based control strategy for grid-connected inverters. The resulting closed-loop current–error dynamics is modeled as a nonlinear dynamical system. We establish two main theoretical results. The first theorem proves ISS of the error dynamics under strong structural conditions on the system. The second theorem leverages the Persidskii-systems framework to relax and simplify these conditions by introducing a Lyapunov function that explicitly incorporates the system nonlinearities. Several illustrative examples and simulations demonstrate the effectiveness and robustness of the proposed control methodology.

\bibliography{bib}

\end{document}